\documentclass[journal]{IEEEtran}
\usepackage{algpseudocode}
\usepackage[linesnumbered,ruled]{algorithm2e}

\usepackage[dvips]{graphicx}
\usepackage{wrapfig}
\usepackage{float}

\usepackage[font=scriptsize]{caption}
\usepackage{cite}
\usepackage[hyphens]{url}
\usepackage{amssymb,amsmath,latexsym}
\usepackage{amsthm}
\newtheorem{theorem}{Theorem}[section]

\newtheorem{lemma}[theorem]{Lemma}

\begin{document}
%
\title{Trust Evaluation Mechanism for User Recruitment in Mobile Crowd-Sensing in the Internet of Things}

\author{Nguyen~Binh~Truong,~\IEEEmembership{Member,~IEEE,}
        Gyu~Myoung~Lee*,~\IEEEmembership{Senior Member,~IEEE,}
				Tai-Won~Um,
        and~Michael~Mackay,~\IEEEmembership{Member,~IEEE}
\thanks{N.B. Truong is with Data Science Institute, Department of Computing, Imperial College London, London,
SW7 2AZ United Kingdom.}
\thanks{G.M. Lee and M. Mackay are with Department of Computer Science, Liverpool John Moores University, Liverpool,
L3 3AF United Kingdom.}
\thanks{T.W. Um is with Department of Information and Communication Engineering, Chosun University, Gwangju, Korea.}
\thanks{*G.M. Lee is the corresponding author}
}

\markboth{IEEE Transaction on Information Forensics and Security, February~2019}%
{Truong \MakeLowercase{\textit{et al.}}: Trust Evaluation Mechanism for User Recruitment in Mobile Crowd-Sensing in the Internet of Things}

\maketitle

\begin{abstract}
Mobile Crowd-Sensing (MCS) has appeared as a prospective solution for large-scale data collection, leveraging built-in sensors and social applications in mobile devices that enables a variety of Internet of Things (IoT) services. However, the human involvement in MCS results in a high possibility for unintentionally contributing corrupted and falsified data or intentionally spreading disinformation for malevolent purposes, consequently undermining IoT services. Therefore, recruiting trustworthy contributors plays a crucial role in collecting high-quality data and providing better quality of services while minimizing the vulnerabilities and risks to MCS systems. In this article, a novel trust model called Experience-Reputation (E-R) is proposed for evaluating trust relationships between any two mobile device users in a MCS platform. To enable the E-R model, virtual interactions among the users are manipulated by considering an assessment of the quality of contributed data from such users. Based on these interactions, two indicators of trust called Experience and Reputation are calculated accordingly. By incorporating the Experience and Reputation trust indicators (TIs), trust relationships between the users are established, evaluated and maintained. Based on these trust relationships, a novel trust-based recruitment scheme is carried out for selecting the most trustworthy MCS users to contribute to data sensing tasks. In order to evaluate the performance and effectiveness of the proposed trust-based mechanism as well as the E-R trust model, we deploy several recruitment schemes in a MCS testbed which consists of both normal and malicious users. The results highlight the strength of the trust-based scheme as it delivers better quality for MCS services while being able to detect malicious users. We believe that the trust-based user recruitment offers an effective capability for selecting trustworthy users for various MCS systems and, importantly, the proposed mechanism is practical to deploy in the real world.
\end{abstract}


\IEEEpeerreviewmaketitle

\section{Introduction} \label{INT}
\IEEEPARstart{E}{merging} Internet of Things (IoT) applications and services depend heavily on data collected from sensing campaigns such as sensor networks and crowd-sourcing. Traditional sensor networks deploy sensors in the terrain to acquire data on a variety of aspects of human lives but they have never reached their full potential and been successfully implemented in the real world. This is due to a number of unsolved challenges, such as high installation cost and insufficient spatial coverage \cite{r01}. The new sensing paradigm called Mobile Crowd-Sensing (MCS), which is a sort of crowd-sourcing that leverages built-in sensors and applications in smart mobile devices, has recently emerged as a promising solution for IoT sensing campaigns \cite{r02}. MCS allows increasing numbers of mobile device owners to share sensed data and, in exchange, the owners get incentives for their contributions. The potential for data collected from smart mobile devices are diverse such as local news, noise levels, traffic conditions, and social knowledge. With diversified spatial coverage due to the mobility of large-scale mobile users, MCS is expected to enable a variety of IoT services including public safety, traffic planning, environment monitoring, and social recommendation. This human-powered sensing approach augments the capabilities of existing IoT infrastructures without introducing additional costs, resulting in a win-win strategy for both users and IoT systems.

However, the introduction of MCS also poses some significant challenges such as cross-space data mining, retaining privacy and providing high-quality data \cite{r03}. Low-quality data could lead to numerous difficulties in providing high-quality services or even damage MCS systems. Certain methods have been proposed for improving the quality of data (QoD) in MCS, including estimation and prediction of sensing data, along with statistical processing for identifying and removing outliers in sensing values \cite{r04}. Data selection techniques are also used to filter low-quality or irrelevant data and to generate a high-quality dataset for further processing in IoT services \cite{r05}. Another approach is the use of a recruitment mechanism for selecting trustworthy users who are expected to contribute high-quality data. An appropriate recruitment scheme would therefore not only reduce system costs but also minimize vulnerabilities, risks and potential attacks in MCS systems.

In this article, a novel trust evaluation mechanism called Experience-Reputation (E-R) is proposed for evaluating trust relationships between any two mobile device users in a MCS platform. To establish and evaluate the trust relationships, we utilize our conceptual trust model in the IoT environment called Reputation-Experience-Knowledge (REK), which comprises of the trust indicators (TIs) called Reputation, Experience and Knowledge proposed in \cite{r06, r07}. To employ the E-R mechanism, virtual interactions between service requesters and data contributors are generated when one user requests a MCS service and other users contribute their sensing data to fulfill it. These interactions are then assessed by performing a QoD assessment over the contributed data. Based on these interactions, Experience relationships between service requesters and data contributors are established and calculated. Then, based on all of these Experience relationships between the users, the Reputation of each user is calculated accordingly. Trust relationships between users are finalized by combining the two associated TIs; Experience and Reputation. As a result, the proposed trust-based recruitment scheme examines the trust relationships between a service requester and potential participants in order to select the most trustworthy contributors for a requested service.

To verify the effectiveness of a user recruitment scheme, we propose an evaluation model for the quality of MCS service (QoS) based on the QoD assessment of data contributed to the service. We simulate the trust-based recruitment scheme along with two popular recruitment mechanisms using predictive algorithms in the same MCS testbed for comparison. The results indicate that the trust-based scheme not only provides better QoS for MCS services but also efficiently differentiates between high-quality, low-quality and malicious users. As a result, using the proposed trust evaluation mechanism for recruiting trustworthy data contributors not only prevents adversaries from contributing falsified data and potential attacks but also motivates users to provide high-quality data in order to be recruited in the next sensing task, hence further strengthening the MCS platform.

The main contributions of this article are three-fold:
\begin{itemize}
	\item The E-R trust mechanism for evaluating trust relationships between MCS users consisting of the Experience and Reputation models.
	\item A practical real-world deployable trust-based user recruitment scheme leveraging the QoD assessment and the E-R mechanism.
	\item A simulation for a MCS testbed consisting of three types of user models deploying different user recruitment schemes including our trust-based user recruitment, and an evaluation model for QoS based on QoD assessment.

\end{itemize}

The rest of the article is organized as follows. Section II presents background and related work on the MCS platform and user recruitment schemes. Section III introduces the trust-based MCS system model and components and the following section specifies the proposed trust evaluation mechanism including the Experience and Reputation computational models in detail. Section V describes the simulation scenarios including the testbed and user recruitment algorithms. Section VI presents the outcomes with analysis and discussion. The last section concludes our work and outlines future research directions.
\section{Mobile Crowd-Sensing Background and Related Work} \label{RLW}

\subsection{Background on Mobile Crowd-Sensing in the IoT}
\begin{figure}[ht]
	\centering
	\captionsetup{justification=centering}
	\includegraphics[width=0.5\textwidth]{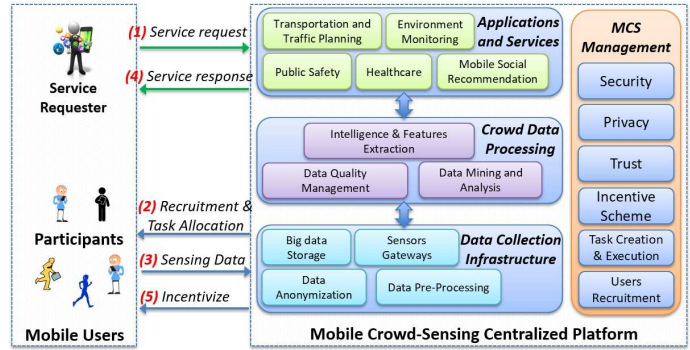}
	\caption{A Centralized MCS Platform Architecture}
	\label{fig1}
\end{figure}
In IoT ecosystems data from various sources such as actuations, sensors, and smart devices are gathered, analyzed and processed to provide ubiquitous and intelligent services \cite{r08, r09}. In this environment, users could contribute to the process through sharing not only data sensed from their own devices’ sensors but also their incidents and knowledge over social networks without the need to pre-allocate sensing devices in the area \cite{r10}, hence saving deployment costs \cite{r11, r12}. This prospect coined the term MCS that has since gained popularity as a promising data acquisition approach for the IoT because of the increasing usage of mobile smart devices. These devices are equipped with many different types of sensors such as Global Positioning System (GPS), accelerometers, a gyroscope, microphone and camera along with advanced features including processing and wireless communications that can efficiently support crowd-sensing processes \cite{r13, r14}. In a MCS platform, heterogeneous information regarding different aspects of human life is collected from mobile devices before being aggregated, analyzed and mined for supporting a variety of IoT applications and services (Fig. \ref{fig1}).

The data acquisition models for a MCS system can be categorized as either opportunistic or participatory \cite{r01}. In optimistic sensing systems, data is automatically collected using a background process, such as reporting speed and GPS coordinates in navigation services while driving. Sensing decisions are application or device-driven, meaning that the involvement of participants is minimal, thus, user recruitment is not necessary. Conversely, in participatory sensing systems, participants agree to a requested sensing task that is dispatched from a centralized MCS platform. Users are explicitly engaged in the sensing process by accepting or rejecting the sensing request; and by actively collecting data such as taking a picture, reporting an available parking lot and manually providing information (illustrated as step (2) and (3) in Fig. \ref{fig1}). Such sensing data can be extracted and directly consumed by end-users for supporting some prompt services or further aggregated in the cloud for large-scale sensing and community intelligent mining \cite{r04}. It is worth to note that in both data acquisition models, the participant trajectories could be revealed by the MCS platform, resulting in the risk of privacy leakage. As a consequence, mobile users may not be enthusiastic to contribute sensing data to the platform even though they get incentives (step (5) in Fig. \ref{fig1}). Privacy-preserving mechanisms for MCS should also therefore be carried out in the MCS platform \cite{r1_4}.

Generally, the life cycle of a MCS system comprises of three phases: `task creation and user recruitment', `task execution', and `data collection and processing' \cite{r15}. More recently, Zhang \textit{et al.} have divided the life cycle into four phases: `task execution', `task assignment', `individual task execution', and `sensing data integration' \cite{r16}, where the `task assignment' phase recruits users and assigns individual sensing tasks for the participants. Nonetheless, the user recruitment scheme plays a key role in the success of any participatory MCS system. The high density of mobile device users, especially in urban areas, allows a MCS system to select only a subset of all available data contributors; and obviously, different user recruitment schemes may lead to different system performances. In order to obtain high quality data, a simple solution is to recruit as many participants as possible \cite{r5_3}. However, collecting data from collocated users may result in data redundancy which cannot further improve the QoD while posing the waste of incentive cost, storage space and imposing network overheads. Therefore, a good recruitment scheme not only selects proper users for providing high-quality data but also allows MCS service providers to manage expenditure by considering incentive costs based on users’ contributions. These MCS systems are tailored to a centralized MCS platform illustrated in Fig. \ref{fig1}, which facilitates major system control operations including the user recruitment.

\subsection{Related Work}
A variety of user recruitment schemes in a centralized MCS platform have been investigated. Reddy \textit{et al.} have proposed a recruitment mechanism in a participatory sensing platform considering some core attributes such as geographic and temporal coverage and user behaviors for defining participant profiles comprising of availability, reputation and cost in their recruitment policy \cite{r17}. Standing on these attributes, Karaliopoulos \textit{et al.} have come up with a deterministic and stochastic mobility model for solving an optimization problem on cost minimization and user location in their recruitment policy \cite{r18}. Lately, other researchers have employed piggyback crowd-sensing techniques for gathering more information from mobile-device owners such as phone call, GPS coordination, and mobile application usages. As a result, these proposed recruitment mechanisms are able to predict geographical coverage and user availability; thus, these mechanisms are capable of determining a minimum number of participants for a sensing task in an energy-efficient recruitment strategy \cite{r19, r20, r21}. For instance, the authors in \cite{r23} have demonstrated a recruitment policy based on statistics of social services usage to compute a 'sociability' metric, indicating the willingness of users to participate in sensing tasks. Wang \textit{et al.} have theoretically leveraged mobile social networks such as Facebook, Twitter and FourSquare as the medium for information sharing and propagation in a novel recruitment platform and proposed two recruitment algorithms. The ultimate goal is to select a near-optimal set of social network users used as seeds (i.e., influenced users) in order to maximize the temporal-spatial coverage of MCS sensing tasks \cite{minor_1}. The authors in \cite{r5_6} have proposed a prediction-based recruitment mechanism considering a factor called `contact probability' indicating whether two MCS users are in the same points of interest (PoIs). They have used a semi-Markov model to determine the probability distribution of the users' arrival time at a PoI to calculate the inter-user contact probability, which is used in a prediction strategy to recruit users with the purpose of lowering data uploading cost. Similarly, Li \textit{et al.} considered a recruitment scheme in a large-scale piggyback MCS system with dynamic and heterogeneous sensing tasks with the aim of minimizing the number of participants while still achieving a stable task coverage \cite{r5_1}. Most of the aforementioned recruitment approaches have the same purposes of developing an energy-efficient and cost-effective recruitment strategy by minimizing the sensing costs for a MCS service provider while guaranteeing certain requirements of requested services such as sensing area coverage. These approaches normally use an auction mechanism for negotiating incentives with mobile-device users \cite{r21, r22}. However, such recruitment mechanisms need to obtain location traces, history of phone calls, and social services personal information, which could pose the risk of serious users privacy leakage. Moreover, the quality of the contributed sensing data from the recruited users is largely neglected. There are multiple factors that affect the recruitment process, and the assurance of high-quality sensing data is of paramount importance.

\begin{itemize}
	\item \textbf{Quality of Data in MCS User Recruitment}
\end{itemize}
Recently, several efforts have proposed to recruit users based not only on time, location and statistical metrics but also on the QoD and the quality of information (QoI). Liu \textit{et al.} have taken the Quality of Information (QoI) requirements of sensing tasks into account for some incentive-based recruitment schemes using a bidding mechanism \cite{r1_1}. However, such schemes only work in a trustworthy environment with no malicious users due to the assumption that the recruited users will always provide data satisfying the QoI requirements for the sensing tasks as in the bid. Li \textit{et al.} also performed statistical analysis on the history of participation in previous sensing tasks for learning and predicting the QoD of the next sensing task \cite{minor_3}. The drawback of this idea is the requirement of calculating the similarity features among sensing tasks in order to recruit high-quality users. The ultimate goal of this work is similar to our work, but our approach is more practical and is not based on the calculation of this similarity. The authors in \cite{r5_2} proposed a participant selection scheme to provide high-QoI satisfaction while minimizing overall energy consumption. The scheme is based on two criteria called the remaining energy level and the `willingness of participation' defined by the rejection probability as the input for a constrained optimization solution. Again, this scheme only works if there is no malicious user who can purposely upload high-quality sensor readings as samples in order to be recruited and then turn out to provide false data to mislead the MCS system. QoI is not only used as a criterion for the user recruitment but also for incentive schemes in MCS systems. For example, the authors in \cite{r5_4} leverage QoI assessment to allocate suitable incentives for data contributors, resulting in a fair incentive mechanism.

\begin{itemize}
	\item \textbf{Reputation and Trust in MCS User Recruitment}
\end{itemize}
In order to deal with the presence of malicious users, reputation can be used as an indicator to perceive trustworthy participators in MCS sensing tasks on the assumption that regular users and adversaries behave differently. Kantarci \textit{et al.} have proposed a reputation-based MCS management approach adopting the M-Sensing auction approach \cite{r24} in which a statistical reputation is taken into account \cite{r25}. This statistical reputation is simply the percentage of true sensor readings over total readings. Pouryazdan \textit{et al.} have further employed a vote-based approach using a social network for evaluating users’ reputation \cite{r26, r27}. In this platform, users who have recently participated in a common sensing task form a community. All members of the community will then vote on the reputation of a newly joining user based on their similarity on sensor readings. The same authors have also considered a vote-based mechanism implementing a Subgame Perfect Equilibrium (SPE) and gamification techniques based on the calculation of users' reputation in the three-step recruitment process for improving the platform utility. The reputation scores are used as the core attributes for recruitment and incentivizing users in sensing tasks in \cite{minor_5, minor_6}. Nevertheless, such reputation-based recruitment schemes have unintentionally claimed the reputation is trust and have used the reputation on its behalf. In reality, reputation is one of several TIs partially affecting trust, but should not be confused with trust itself \cite{r06}. Moreover, the mechanisms used in such approaches are either too simple \cite{r25, r29}, based only on statistical sensor readings, or impractical assumptions \cite{r26, r27, minor_5, minor_6}. For instance, if two users join in the same sensing task, then there will be an interaction between them and they will get connected and directly interact with each other. Another assumption is that any user has the right to access all previous readings of other users in the same community for making up their votes. This results in the unfeasible deployment of these mechanisms in the real world. The authors in \cite{r1_2} have proposed a dynamic trust-based framework for recruiting suitable mobile users that provide high-quality sensing data on time. In that paper, an overall trust degree is calculated for selecting trustworthy users by aggregating from three factors: Direct Trust, Feedback Trust and an Incentive Function. The final goal is similar to our research work, however, the drawback of this approach is that it requires feedback from task recruiters for the Feedback Trust as well as to keep track of non-cooperative behaviors of mobile users for the Incentive Function. Restuccia \textit{et al.} has summarized recent research about developing a framework for discovering trust in MCS \cite{r5_5}. They have furthermore discussed current challenges and different approaches for evaluating trust through a collection of trust indicators.

Given this state-of-the-art, we propose a trust evaluation mechanism that can be effectively used to recruit trustworthy users while still being practically deployable for real-world services.
\section{E-R TRUST MECHANISM IN MCS PLATFORM: MODEL AND SYSTEM COMPONENTS} \label{ERTrust}
This section explores a MCS system model and scenarios, then introduces the E-R trust evaluation mechanism and its components deployed on top of a centralized MCS platform.

\subsection{MCS System Model and Scenario}
In a MCS platform, users share and provide data from their smart devices through being physically close (direct sensing model) or through a centralized MCS platform (indirect sensing model) \cite{r30}. In the direct sensing model, direct interactions exist between a requester and a provider such that sensing data is transmitted in a peer-to-peer manner. This sensing model uses a variety of wireless communication technologies such as Wi-Fi direct, ZigBee, Near-Field Communication (NFC) and Bluetooth over a social platform that operates among nearby smart device users \cite{r31, r32}. In the indirect sensing model, a requester and a provider indirectly interact via a centralized MCS platform. In this model, users can upload and obtain data to and from a cloud server through wide-range communication technologies such as Wi-Fi and 3G/4G. The indirect sensing model adopts the well-known service-oriented approach model called Sensing as a Service (S2aaS) \cite{r33}. Melino \textit{et al.} have further developed a Cloud-based SaaS model designated for MCS systems called Mobile Crowd-Sensing as a Service (MCSaaS) \cite{r34}.

Nevertheless, in any MCS model, a user can be either a requester that asks for a service or a data provider that collects and delivers data being used by another service; thus MCS users are directly or indirectly interacting with each other. This introduces either a `direct' or an `indirect' relationship between a service requester and a data provider depending on the sensing model deployed in a MCS system. In this article, we consider MCS systems that adopt the indirect sensing model with a participatory data acquisition style, which is overwhelmingly the most common in real-world usage. For such a system model, there is a centralized MCS cloud platform that handles and operates all the MCS processes including data collection and processing, task creation and execution; and the user recruitment and incentive schemes as illustrated in Fig. \ref{fig1}.

\subsection{E-R Trust Mechanism in the MCS Platform}
Trust can be considered as the `belief' of a trustor that the trustee will perform a task as the trustor’ expects. Trust plays an important role in supporting participants to overcome the perception of uncertainty and risks when making a decision \cite{r06}. In the MCS context, trust can be utilized to predict whether a mobile device user (i.e., the trustee) is going to provide high-quality data for a service requested by a service requester (i.e., the trustor). To establish and evaluate trust relationships between service requesters and data contributors, the REK trust model proposed in \cite{r06, r07, r35} is employed.

As depicted in Fig. \ref{fig2}, trust is comprised of three TIs called Reputation, Experience and Knowledge. Knowledge is identified as `direct trust' and evaluated by inferring trustees’ characteristics considering the trust context \cite{r06}. In the MCS context, Knowledge is constituted from a variety of attributes such as availability, the mobility model, GPS coordination and geography coverage. These attributes specify criteria for user ability and eligibility for fulfilling crowd-sensing campaigns. Experience and Reputation in contrast are identified as “indirect trust” and are quantified by accumulating previous interactions between mobile device users. Experience is a relationship between two users reflecting the personal perception of a trustor on a trustee. Reputation is the property of a user indicating the global consciousness of that user by considering all personal perceptions toward it \cite{r06}.

\begin{figure}[ht]
	\centering
	\captionsetup{justification=centering}
	\includegraphics[width=0.45\textwidth]{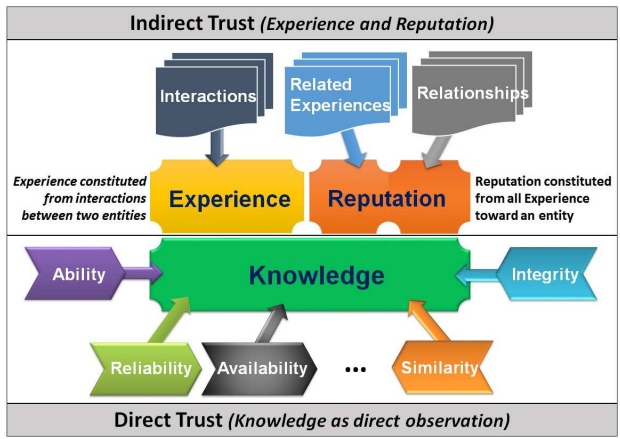}
	\caption{Trust Indicators and Attributes in the REK Trust Model}
	\label{fig2}
\end{figure}

Knowledge assessment requires various information from mobile device users that imposes critical privacy concerns in this context. Moreover, some information is challenging to retrieve or is not practical to implement in real-world scenarios \cite{r06}. For those reasons, we simplify the REK model which we will now call E-R that relies only on two indicators; Experience and Reputation. Knowledge is neglected in the E-R model, but some information could play a supplemental role in strengthening the evaluation of trust. As illustrated in Fig. \ref{fig3}, the E-R trust components are integrated in a centralized MCS cloud platform that establishes and manages virtual interactions between mobile-device users. An indirect interaction occurs after each sensing task is accomplished; and the interaction value is calculated based on the QoD provided to the MCS system (from data providers) and feedback (from service consumers). Experience between any two users is established and updated by an aggregation model on the virtual interactions. Based on all Experiences between users, the Reputation of each user is calculated accordingly. Finally, the value of a trust relationship is calculated by aggregating the Experience and Reputation. Detailed calculation models for the Experience, Reputation and trust are presented in Section IV.

\begin{figure}[ht]
	\centering
	\captionsetup{justification=centering}
	\includegraphics[width=0.48\textwidth]{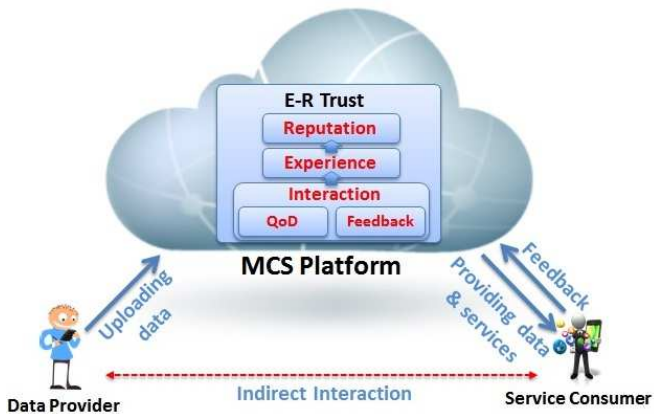}
	\caption{E-R Trust Mechanism in a Centralized MCS platform}
	\label{fig3}
\end{figure}

\subsection{Quality of Data Assessment}
The aim of MCS systems is to extract useful knowledge and intelligence from sensing data for delivering smart services; and to achieve this, high QoD must be ensured \cite{r36}. Low-quality data might cause numerous problems such as deception in decision making, consumer dissatisfaction and distrusting the system \cite{r37}. Well-known research works have pointed out that QoD consists of evaluating measurable properties that represent certain aspects of the data \cite{r37, r38}, and some data can be identified as high quality based on the measurements of these dimensions \cite{r37}. Six data quality dimensions are specified by Askham \textit{et al.} in \cite{r38} and have been widely accepted, namely Accuracy, Completeness, Consistency, Timeliness, Uniqueness, and Validity. Detailed analysis and measurement methodologies for the six dimensions have also been proposed in related articles. Therefore, based on the system requirements, context, and system goals these dimensions can be taken into consideration for the QoD assessment \cite{r40, r41}.

\begin{figure}[ht]
	\centering
	\captionsetup{justification=centering}
	\includegraphics[width=0.48\textwidth]{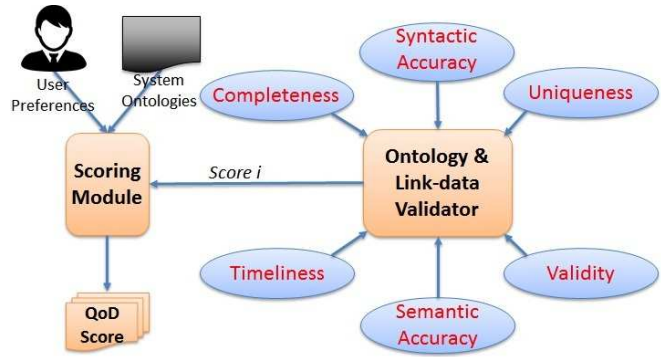}
	\caption{QoD Monitoring Module for traffic and parking sensors in the Wise-IoT project}
	\label{fig4}
\end{figure}

We have utilized the QoD calculation mechanisms in \cite{r37, r38} for measuring live data streaming QoD from traffic sensors and parking sensors deployed in Santander City Center, Spain as a result of the Wise-IoT\footnote{\url{http://wise-iot.eu/en/home}} project. As the data is presented in semantic form, we have proposed two further novel dimensions called Syntactic Accuracy and Semantic Accuracy in the QoD assessment \cite{r42}. These two dimensions are suitable for checking data syntax and semantics from live information produced by the sensors (Fig. \ref{fig4}) using predefined data quality rules as well as the ontology validating rules developed by EGM\footnote{\url{http://www.eglobalmark.com}} \cite{r42}. We believe this mechanism can be reused here for evaluating sensing data in a MCS platform because the underlying theoretical and practical QoD assessments are identical.

\subsection{User Feedback}
QoD is the most important indicator of how contributors fulfill an assigned sensing task but it may not be sufficient alone because QoD scores do not completely reflect the level of consumers’ satisfaction with the service provider. In this regard, feedback can complement the assessment of to what extent a service provider has accomplished a requested service. Feedback can be both implicit and explicit; and may or may not require human participation. Feedback could be obtained by directly asking customers to give opinions after a service has been provided. This approach has been used in many e-commerce services such as eBay, Amazon and Airbnb, which requires huge effort to attract users to anticipate; and opinions are sometimes biased. The implicit approach is based on calculation models with some predefined criteria to estimate the outcome, which normally do not require a human participant. For example, this has been applied in some networking protocols as an ACK message to indicate whether a packet or a file is transmitted successfully or unsuccessfully \cite{r45}.

However, this type of user feedback is out of scope of this article. In the E-R trust component we neglect the feedback mechanism at this stage and thus indirect interactions between users rely on QoD scores only. However, user feedback could be an important component for improving the quality of IoT services and we will consider it as part of further work.
\section{E-R TRUST EVALUATION MODEL} \label{ERTrustEva}

In this section, the mathematical calculation models for the E-R trust mechanism are described in detail.

\subsection{Experience Model}
Experience is an asymmetric relationship between two entities built up from previous interactions reflecting to what extent a trustor trusts a trustee. After each interaction, awareness between the trustor and the trustee is supposed to improve, and Experience should be maintained to correctly indicate the relationship between the two (illustrated in Fig. \ref{fig5}).

\begin{figure}[ht]
	\centering
	\captionsetup{justification=centering}
	\includegraphics[width=0.3\textwidth]{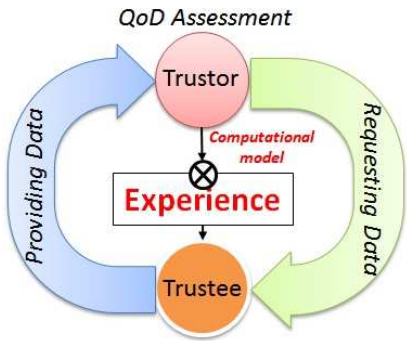}
	\caption{Experience computation model based on feedback mechanism}
	\label{fig5}
\end{figure}

The proposed Experience model for MCS systems follows human relationships investigated in sociological literature \cite{r46, r47}. That is, Experience increases due to cooperative interactions and decreases by uncooperative interactions. Experience also decays if no interactions occur after a period of time. The amount of the increase, decrease and decay depends on the intensity of interactions, interaction scores, and current Experience value. Therefore, Experience can be modeled using mathematical models as follows with the notations denoted in Table \ref{tb1_1}:

\begin{table}[h]
	\centering
	\caption{NOTATIONS USED IN THE EXPERIENCE MODEL}
	\label{tb1_1}
	\begin{tabular}{c|c}		
			\textbf{Notations} & \textbf{Description} \\ [0.5ex]
		\hline\hline
			$Exp_t$ & Experience value at the time $t$ \\
		\hline
			$max_{Exp}$ & Maximum value of Experience, normally set to 1 \\
		\hline
			$min_{Exp}$ & Minimum value of Experience, normally set to 0 \\
		\hline
			${Exp_0}$ & Initial Experience value at the bootstrap $t=0$ \\
		\hline
			$\vartheta_t$ & Interaction value (i.e., QoD score) at the time $t$ \\
		\hline
			$\alpha$ & Maximum Increase value, $0<\alpha<max_{Exp}$ \\
		\hline
			$\beta$ & Rate of the Decrease, normally $\beta>1$ \\
		\hline
			$\theta_{co}$ & Cooperative Threshold for the $\vartheta_t$ \\
		\hline
			$\theta_{unco}$ & Uncooperative Threshold for the $\vartheta_t$\\
		\hline
			$\delta$ & Minimum Decay value \\
		\hline		
		  $\gamma$ & Decay Rate\\
	\end{tabular}	\\[0.1ex]
\end{table}

\begin{itemize}
	\item \textbf{Increase Model (due to cooperative interactions)}
\end{itemize}
A cooperative interaction is when $\vartheta_t \geq \theta_{co}$. The Increase function is modeled using a linear difference equation as follows:

\begin{equation}
	\label{eq_1}
	Exp_t = Exp_{t-1} + \vartheta_t {\Delta}Exp_t
\end{equation}

\begin{equation}
	\label{eq_2}
	{\Delta}Exp_t = \alpha(1 - \frac{Exp_{t-1}}{max_{Exp}})
\end{equation}

\begin{itemize}
	\item \textbf{Decrease Model (due to uncooperative interactions)}
\end{itemize}
An uncooperative interaction is when the QoD score $\vartheta_t \leq \theta_{unco}$ threshold. The Decrease function is modeled as follows:

\begin{equation}
	\label{eq_3}
	Exp_t = Max(min_{Exp}, Exp_{t-1} - (1 - \vartheta_t)\beta {\Delta}Exp_t)
\end{equation}
Where ${\Delta}Exp_t$ is already determined by (\ref{eq_2}).

\begin{itemize}
	\item \textbf{Decay Model (due to no or neutral interactions)}
\end{itemize}
Experience TI decays if there is no transaction after a period of time or the interactions are neutral (i.e., $\theta_{unco} < \vartheta < \theta_{co}$). The Decay function is proposed as follows:

\begin{equation}
	\label{eq_4}
	Exp_t = Max(Exp_0, Exp_{t-1} - \Delta{Decay_t})
\end{equation}

\begin{equation}
	\label{eq_5}
	\Delta{Decay_t} = \delta{(1 + \gamma - \frac{Exp_{t-2}}{max_{Exp}})}
\end{equation}

\subsection{Analysis and Discussion for Experience Model}

As we are imitating the relationships seen in human society, it is expected that Experience TI is accumulated from cooperative interactions; and $Exp_{t+1}$ depends on both QoD score $\vartheta_t$ and current value $Exp_t$. Also, a strong relationship should require more and more cooperative interactions to attain. Considering the trust evaluation in which trust values and QoD scores are in the range $(0,1)$, Experience values should be normalized to the range $(0,1)$, thus we set $max_{Exp} = 1$, $min_{Exp} = 0$, and $0 < Exp_0 < 1$. It is obvious that the increase model defined in (\ref{eq_1}) and (\ref{eq_2}) is incremental; and the increase value from time $t$ to time $t+1$ $\vartheta_t\times{\Delta}Exp_{t+1} > 0$ is relatively large when the current value of $Exp_t$ is small and vice versa (considering the same interaction value $\vartheta$), meaning that higher $Exp$ gets more difficult to achieve. 

\begin{lemma}
	The proposed increase function $Exp$ is \textit{always less than $1$} and \textit{asymptotic to $1$}.
\end{lemma}
 
\begin{proof}
	From (\ref{eq_1}) and (\ref{eq_2}) with $max_{Exp} = 1$, the $Exp$ function can be re-written as:
	\begin{equation}
		\label{eq_1x}
		Exp_t = Exp_{t-1} + (1 - Exp_{t-1})\vartheta_t\ \alpha
	\end{equation}

	Subtracting both sides of (\ref{eq_1x}) from $1$:

	\begin{align}
		1 - Exp_t &= 1 - (Exp_{t-1} + (1 - Exp_{t-1})\vartheta_t\ \alpha) \nonumber \\
							&= (1 - Exp_{t-1})(1-\vartheta_t\ \alpha) \nonumber \\
							&= (1 - Exp_{t-2})(1-\vartheta_t\ \alpha)(1-\vartheta_{t-1}\ \alpha) \nonumber \\
							&= ... \nonumber \\
							&= (1 - Exp_0)\prod_{i=1}^{t} (1 - \vartheta_i\ \alpha)
		\label{eq_2x}
	\end{align}
	
	According to (\ref{eq_2x}), because $0 < Exp_0$, $\vartheta_t$, and $\alpha < 1$, $1 - Exp_t > 0$; in other words, $ 0 < Exp_t < 1\ \forall{t}$.
	
	Moreover, because $\vartheta_i \geq \theta_{co}; \forall{i} \in \{1,..,t\}$, we have:
	\begin{equation}
		\label{eq_3x}
		0 < 1 - Exp_t \leq (1 - Exp_0)(1 - \theta_{co}\ \alpha)^t
	\end{equation}

Because $0<\theta_{co}$, $\alpha$, and $Exp_0<1$ are three pre-defined parameters, thus:
	\begin{equation}
		\label{eq_4x}
		\lim_{t\to\infty} (1 - Exp_0)(1 - \theta_{co}\ \alpha)^t = 0
	\end{equation}
	
Applying the Squeeze theorem on (\ref{eq_3x}) and (\ref{eq_4x}), we have: $\lim_{t\to\infty} (1 - Exp_t) = 0$. Therefore, the increase of $Exp$ is \textit{asymptotic to $1$}.

\end{proof}

As with the Increase function, the Decrease function in (\ref{eq_3}) is decremental and the decrease value depends on both the current value of $Exp_t$ and the uncooperative $\vartheta_t$ QoD score. It is worth to note that the Decrease rate $\beta$ should be greater than $1$ because a strong relationship (i.e., high $Exp$ value) is difficult to gain but easy to lose (e.g., $\beta = 2$ means that the $Exp$ value decrease due to uncooperative interactions is \textit{twice} compared to the amount gained in the corresponding cooperative interaction). The Decrease function also ensures that strong relationships are more resistant to uncooperative interactions whereas weak relationships are severely damaged.

Regarding to the Decay function, $\delta$ is the minimal decay value which guarantees that even strong relationships still get decreased; and $\gamma$ is the decay rate. In sociology, relationships between people decay over time if participants do not interact, although the decay rates are different depending on the strength of the relationships \cite{r48}. Similarly, the proposed decay model shows that relationships require periodic maintenance, but strong ones tend to persist longer even without reinforcing cooperative interactions. As can be seen in (\ref{eq_4}), the decay value is assumed to be inversely proportional to the current Experience value, thus strong relationships exhibit less decay than weak ones.

\subsection{Reputation Model}
Reputation is a property of a user reflecting the overall opinion of a community about that user. In the MCS environment, especially in urban scenarios with a large number of mobile users, only small numbers of users have already interacted with others, resulting in a very high possibility that a service requester (i.e., the trustor) and a data provider (i.e., the trustee) are new to each other, thus no prior Experience relationship exists between the two. The Reputation of the trustee, therefore, is a vital indicator for the trust evaluation.

As Reputation is an overall opinion, the calculation for the reputation of a user $U$, denoted as $Rep(U)$, needs to take all users $\forall{i}$ that have prior Experience with $U$ into consideration. Intuitively, Reputation can be quantified using a graph analysis algorithm on the Experience relationship graph, which is somewhat similar to the Google PageRank \cite{r49} and the weighted PageRank \cite{r50} approaches. The difference from the two previous models is that each user $i$ contributes differently to $Rep(U)$, in either a positive or negative manner, depending on both $Exp(i,U)$  (i.e., the Experience from $i$ toward $U$) and the user’s reputation (i.e., $Rep(i)$).

To come up with the new model for Reputation, we modify the PageRank models proposed in \cite{r49, r50} by classifying the Experience relationships into two sub-groups: Positive Experiences (i.e., $Exp>\theta$) and Negative Experiences (i.e., $Exp<\theta$) where $\theta$  is a predefined threshold. Let $N$ be the number of users in a MCS system, and $d$ is a damping factor ($0<d<1$) as defined in the standard PageRank \cite{r6_18}. Then, the Reputation model is proposed as a composition of the two components \textit{Positive Reputation} and \textit{Negative Reputation} as follows:

\begin{itemize}
	\item \textbf{Positive Reputation}
\end{itemize}
The positive reputation can be calculated as follows:
\begin{equation}
	\label{eq_6}
	Rep_{Pos}(U) = \frac{1-d}{N} + d(\sum_{\forall{i}}Rep_{Pos}(i)\times\frac{Exp(i, U)}{C_{Pos}(i)})
\end{equation}
Where: $C_{Pos}(i) = \sum_{Exp(i,j)>\theta}{Exp(i,j)}$ is the sum of all positive Experience from user $i$.

\begin{itemize}
	\item \textbf{Negative Reputation}
\end{itemize}
The negative reputation can be calculated as follows:
\begin{equation}
	\label{eq_7}
	Rep_{Neg}(U) = \frac{1-d}{N} + d(\sum_{\forall{i}}Rep_{Neg}(i)\times\frac{1- Exp(i, U)}{C_{Neg}(i)})
\end{equation}
Where: $C_{Neg}(i) = \sum_{Exp(i,j)<\theta}{(1- Exp(i,j))}$ is the sum of all compliment of negative Experience from user $i$.

\begin{itemize}
	\item \textbf{Overall Reputation}
\end{itemize}
Finally, the overall reputation is the combination of the two positive and negative reputations:
\begin{equation}
	\label{eq_8}
	Rep(U) = max(0, Rep_{Pos}(U) - Rep_{Neg}(U))
\end{equation}

\subsection{Mathematical Analysis for Reputation Model}

According to the proposed model, the reputation of a user is recursively calculated from other users' reputations and the corresponding Experience relationships; consequently reputations of all $N$ users (forming a $N$-$element$ vector denoted as $Rep$) in a MCS platform are correlated with each other.  Therefore, this $Rep$ vector might not exist due to the correlations among $N$ users' reputations; or the $Rep$ vector might be ambiguous (i.e., not unique: a user might have more than one reputation value) which is not reasonable.
\begin{lemma}
	The reputation vector $Rep$ calculated by the proposed reputation model \textit{exists} and is \textit{unique}.
\end{lemma}

\begin{proof}
Regarding (\ref{eq_6}), let $M$ be the $N \times N$ diagonal matrix where the diagonal element $m_i = C_{Pos}(i) \forall{i} \in \{1,..,N\}$. Let $Exp_{Pos}$ be a $N \times N$ matrix that:
\begin{equation}
	\label{eq_9}
	Exp_{Pos}(i,j) = 
		\begin{cases}
			Exp(i,j) & \text{if } Exp(j,i) \geq \theta \\
			0 & \text{if } Exp(j,i) < \theta \\
		\end{cases}
\end{equation}

Let $Rep_{Pos}$ be the positive reputation vector consisting of $N$ elements $Rep_{Pos}(i) \forall{i} \in \{1, ..,N\}$. Then, (\ref{eq_6}) can be expressed in matrix notation as follows:
\begin{equation}
	\label{eq_10}
	Rep_{Pos} = (\frac{1-d}{N}{\times}E + d{Exp_{Pos}}{\times}M^{-1}){\times}Rep_{Pos}		
\end{equation}
where $E$ is a $N{\times}N$ matrix of $1s$. Let us define:

\begin{equation}
	\label{eq_11}
	A = \frac{1-d}{N}{\times}E + d{Exp_{Pos}}{\times}M^{-1}
\end{equation}

Thus, (\ref{eq_10}) can be rewritten as:
\begin{equation}
	\label{eq_12}
	Rep_{Pos} = A{\times}Rep_{Pos}		
\end{equation}

Therefore, $Rep_{Pos}$ is the $eigenvector$ of matrix $A$ with $eigenvalue = 1$. We now prove that the $eigenvector$ $Rep_{Pos}$ of the matrix $A$ \textit{exists} and is \textit{unique}. Equation (\ref{eq_9}) and (\ref{eq_10}) is reminiscent of the stationary distribution of a Markov chain which moves among the set of $N$ states from $1$ to $N$ with the $N{\times}N$ transition matrix $P$ where $P($go from state $i$ to state $j$$) = P(i,j)$. 

Let us consider a discrete-time Markov chain defined by a set of states as the $N$ entities and a transition probability matrix $P = A^T$:
\begin{equation}
	\label{eq_13}
	P(i,j) = A^T(i,j) = A(j,i) = \frac{1-d}{N} + d\frac{Exp_{Pos}(j,i)}{m(j)}
\end{equation}

Consequently, the Markov chain can be defined as following:
\begin{equation}
	\label{eq_14}
	P(i,j) = 
		\begin{cases}
			\frac{1-d}{N} + d\frac{Exp_{Pos}(j,i)}{m(j)} & \text{if } Exp(j,i) \geq \theta \\
			\frac{1-d}{N} & \text{if } Exp(j,i) < \theta \\
		\end{cases}
\end{equation}

Fortunately, this turns to a model of \textit{random surfer} with \textit{random jumps} as in the edge-weighted PageRank model \cite{r6_29}. This leads us to show the Markov chain is strongly connected, and the $Rep_{Pos}$ vector, which is the \textit{stationary distribution} of the Markov chain, is \textit{unique} \cite{r6_18, r6_26, r6_28}.

Similarly, the $Rep_{Neg}$ vector from (\ref{eq_7}) \textit{exists} and is \textit{unique}. As a consequence, the overall reputation vector $Rep$ defined in (\ref{eq_8}) also \textit{exists} and is \textit{unique}.

\end{proof}

\subsection{Final Trust Value}
A trust value is an aggregation of both the Experience and Reputation values. There are a variety of techniques for combining the two TIs such as Bayesian neutron networks, fuzzy logic, and machine learning depending on the specific use-cases and individual users’ preferences. A simple weighted sum for calculating a final trust value between trustor A and trustee B is used as follows:

\begin{equation}
	\label{eq_15}
	Trust(A, B) = w_1 Rep(B) + w_2 Exp(A, B)
\end{equation}
Where $w_1, w_2>0$ are weighting factors satisfying $w_1+w_2=1$. The weighting factors can be autonomously tuned using different techniques such as machine learning and semantic reasoning.
\section{SIMULATION TESTBED AND USER RECRUITMENT SCHEMES} \label{STU}
This section presents a MCS testbed in which the trust-based user recruitment is simulated along with two other schemes called Average and Polynomial Regression predictive models \cite{r51}.

\subsection{User Models in MCS}
Some statistics and analysis were carried out on QoD scores in a real-time data stream collected from traffic sensors\footnote{\url{https://mu.tlmat.unican.es:8443/v2/entities?limit=1&type=ParkingSpot}} and parking sensors\footnote{\url{https://mu.tlmat.unican.es:8443/v2/entities?limit=1&type=TrafficFlowObserved}} deployed in the city of Santander, Spain as part of the Wise-IoT project. Histograms of QoD from various sensors were analyzed and normalized in the range $(0, 1)$. Based on this histogram, we have observed that the QoD score distribution from any sensor nicely fits to the Beta probability distribution family. By using a Beta parameter estimation mechanism, we categorize users in a MCS system into three groups based on their QoD score distribution as follows:

\begin{figure}[ht]
	\centering
	\captionsetup{justification=centering}
	\includegraphics[width=0.48\textwidth]{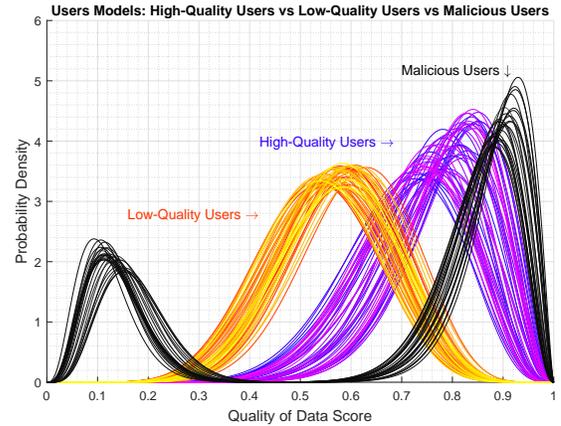}
	\caption{User Models in MCS systems}
	\label{fig6}
\end{figure}

\begin{itemize}
	\item \textbf{High-Quality Users}
\end{itemize}
High quality users consistently produce high QoD scores in most sensing tasks. Based on the statistical information, QoD scores from a high-quality user distribute in the interval $(0, 1)$ but the highest distribution is in the range $(0.75-0.85)$. QoD scores from a high-quality user follow a unimodal Beta distribution with two positive shape parameters $Beta(\alpha_{high}, \beta_{high})$ satisfying $10<\alpha_{high}<15$ and $3<\beta_{high}<5$. The probability density function (PDF) of the Beta distributions for 50 high-quality users are shown in Fig. \ref{fig6}.

\begin{itemize}
	\item \textbf{Low-Quality Users}
\end{itemize}
Low-quality users consistently produce average or below-average QoD scores in most of the sensing tasks. QoD scores are in the $(0, 1)$ interval but mostly fall in the range $(0.5-0.65)$. Similar to high-quality users, QoD scores from a low-quality user follow a unimodal Beta distribution with the two positive shape parameters $Beta(\alpha_{low}, \beta_{low})$ satisfying $9<\alpha_{low}<12$ and $7<\beta_{low}<9$. The PDF of the Beta distribution for 50 low-quality users are depicted in Fig. \ref{fig6}.

\begin{itemize}
	\item \textbf{Intelligent Malicious Users}
\end{itemize}
Even though no data from malicious smart devices was collected, a feasible intelligent malicious user might follow the behaviors below:

\begin{itemize}
	\item Normally produces very high QoD scores in order to pose as a strong candidate for recruitment schemes.
	\item Unpredictably and intentionally produces very low-quality data once the user is recruited in a sensing task to destroy a targeted MCS service. The service will be heavily damaged if the data is used for fulfilling requested services.
\end{itemize}
According to the above description, the malicious user model follows a bi-modal Beta distribution. Thus, firstly we define two Beta distribution models, one for very high QoD scores $Beta(\alpha_{mhigh}, \beta_{mhigh})$, satisfying $18<\alpha_{mhigh}<22$ and $2.5<\beta_{mhigh}<3.5$; and one for very low QoD scores $Beta(\alpha_{mlow}, \beta_{mlow})$, satisfying $4<\alpha_{mlow}<6$ and $25<\beta_{mlow}<35$. Then the two Beta distributions are mixed in order to form the desired bimodal Beta distribution $BiBeta$ using a mixture coefficient parameter $\gamma$ as follows:

\begin{equation}
	\label{eq_16}
	\begin{aligned}
		PDF(BiBeta)={\gamma}PDF(Beta(\alpha_{mhigh}, \beta_{mhigh})) \\
		+(1-\gamma)PDF(Beta(\alpha_{mlow}, \beta_{mlow}))
	\end{aligned}
\end{equation}

Fig. \ref{fig6} also illustrates 25 malicious users with the mixture coefficient $\gamma=0.7$, meaning that the users follow the $Beta(\alpha_{mhigh}, \beta_{mhigh})$ in $70\%$ of the sensing tasks (providing high quality data) and provide very low quality data in $30\%$ of the sensing tasks (i.e., following the $Beta(\alpha_{mlow}, \beta_{mlow})$).

\subsection{QoS Evaluation Model for MCS Services}
To evaluate and compare the effectiveness between different user recruitment schemes in the performance of MCS services, a QoS evaluation model is proposed. Low-quality data lowers system efficiency and misleads system operations that directly leads to customer dissatisfaction \cite{r52}. Low-quality data also increases system operational overheads and cost; and imposes vulnerabilities and risks to the system. Some QoS evaluation models for IoT services have been proposed, taking into consideration different factors at various layers of the IoT infrastructure \cite{r54}; and the QoD is one of the pivotal factors in the evaluation of QoS for MCS services.

Considering a service request $R$ that comprises of $T$ sensing tasks $ST_R(i); \forall{i} \in \{1,..,T\}$; each sensing task $ST_R(i)$ is fulfilled by $P_i$ participants providing $P_i$ datasets with $QoD_{ST_R(i)}(j); \forall{j} \in \{1,..,P_i\}$ respectively. The QoS for the service R is calculated as follows:

\begin{equation}
	\label{eq_17}
	QoS(request R) = \frac{T}{\left|\log({\prod_{i=1}^{T}QoD_{ST_R(i)}})\right|}
\end{equation}
\begin{equation}
	\label{eq_18}
	QoD_{ST_R(i)} = \frac{\sum_{j=1}^{P_i}QoD_{ST_R(i)}(j)}{P_i}
\end{equation}

Equation (\ref{eq_17}) depicts that the QoS of the service request $R$ is proportional to the QoD scores of each the sensing task $QoD_{ST_R(i)}; \forall{i} \in \{1,..,T\}$, represented by the product of the natural logarithm of these scores. The $QoD_{ST_R(i)}$ score of the sensing task $ST_R(i)$ is calculated by taking the average of the QoD scores from the $P_i$ contributors associated to the sensing task. This is because contributors in the same sensing task are normally required to collect the same sort of data; such redundant datasets are then filtered and pre-processed to retrieve a high-quality dataset before processing and mining. However, the number of participants in each sensing task should be small enough in order to not incur significant computation and storage overheads. Nevertheless, user recruitment plays a crucial role in providing high-quality services because even in a sensing task fulfilled by many participants, some attackers providing extremely low QoD data could result in massive damage to MCS services.

\subsection{Trust-based, Average, and Polynomial Regression User Recruitment Schemes}
Generally, all three recruitment schemes have the same purpose of recruiting mobile device users that are expected to provide high QoS scores for sensing tasks in a MCS service request. The algorithms to recruit users in the three schemes rely only on QoD scores of sensing data contributed by users who have been recruited in previous sensing tasks. The Trust-based recruitment scheme uses trust relationships between a service requester and other users for recruiting participants. The Average-QoD and Polynomial Regression-QoD schemes use the two popular predictive schemes; namely Average and Polynomial Regression, respectively, for predicting the QoD scores, and recruiting users who are likely to provide the highest QoD scores for the next sensing task accordingly.

For the comparison among the recruitment schemes, all of the algorithms have the same inputs consisting of $N$ Users, $M$ Service Requests, and associated sensing tasks and the same output as the QoS score for the $M$ requested services:
\begin{algorithm}
	\SetKwInOut{Input}{Input}
	\SetKwInOut{Output}{Output}
	\Input{$N$ Users. $M$ Service Requests $R(i)$ $\forall{i} \in \{1,..,M\}$. Each $R(i)$ requires $T_i$ Sensing Tasks $ST_{R(i)}(j)$ $\forall{j} \in \{1,..,T_i\}$ and $\forall{i} \in \{1,..,M\}$. Each $ST_{R(i)}(j)$ is fulfilled by $P_{ij}$ participants $\forall{j} \in \{1,..,T_i\}$.}

	\Output{$QoS$ scores for the $M$ Service Requests}
	\caption{Inputs and Outputs for User Recruitment Algorithms}\label{inout_alg}
\end{algorithm}

Then, the three algorithms are demonstrated in mathematical-style pseudo-code as follows:

\begin{itemize}
	\item \textbf{Trust-based User Recruitment scheme}
\end{itemize}
This scheme establishes and maintains trust relationships between users based on the E-R trust model proposed in Section IV and recruits users with the highest trust values with a service requester. As can be seen in Algorithm \ref{trustbased}, it firstly initiates the matrices EXP, REP and TRUST for keeping track of Experience relationships, Reputation values, and Trust relationships for $N$ users (line \#1). The output at the beginning state is set to 0 (line \#2). For each request $R(i)$ from a user $U(i)$ and for each sensing task $ST_{R(i)}(j)$, the algorithm recruits participants that have the highest trust values with $U(i)$ (line \#5). When the sensing task has been accomplished, the algorithm calculates the QoD score for the sensing data collected from the recruited users and updates EXP, REP and TRUST accordingly (line \#6 to line \#9). Finally, the output is updated by adding the QoS score of the requested service $R(i)$ (line \#11).

\begin{algorithm}
	\SetKwFunction{Update}{Update}
	\SetKwFunction{Recruit}{Recruit}
	\SetKwFunction{QoD}{QoD}
	\SetKwFunction{QoS}{QoS}

	\textbf{Initialization}
		TRUST[][], EXP[][], REP[]; \;
		$out$ = 0; \;

	\BlankLine
	\ForEach{request $R(i)$ from user $U(i)$}{
		\ForEach{sensing task $ST_{R(i)}(j)$}{	
				\Recruit{$P_{ij}$ users with highest TRUST[$P_{ij}$ users][$U(i)$]}\;
				\QoD{Sensing data from $P_{ij}$ users}\;
				\Update{EXP[$U(i)$][$P_{ij}$ users]}\;
				\Update{REP[]}\;
				\Update{TRUST[][]}\;
		}
		$out$ $\leftarrow$ $out$ + \QoS{$R(i)$};
	}
	\textbf{Return} $out$
	\caption{Trust-Based Recruitment Algorithm}\label{trustbased}
\end{algorithm}

\begin{itemize}
	\item \textbf{Average-QoD User Recruitment scheme}
\end{itemize}

This scheme maintains a list of the average QoD scores for $N$ users and recruits participants with highest average QoD scores. As can be seen in Algorithm \ref{avgbased}, it initiates the AVG matrix for keeping track of the average QoD scores for $N$ users (line \#1). The output at the beginning state is set to 0 (line \#2). For each request $R(i)$ from a user $U(i)$ and for each sensing task $ST_{R(i)}(j)$, the algorithm simply recruits participants with the highest average QoD score (line \#5). When the sensing task has been accomplished, the algorithm calculates the QoD score for the sensing data collected from the recruited users (line \#6) and updates the AVG matrix accordingly (line \#7). Finally, the output is updated by adding the QoS score of the requested service $R(i)$ (line \#9).

\begin{algorithm}
	\SetKwFunction{Update}{Update}
	\SetKwFunction{Recruit}{Recruit}
	\SetKwFunction{QoD}{QoD}
	\SetKwFunction{QoS}{QoS}
	
	\textbf{Initialization}
		AVG[]; \;
		$out$ = 0; \;

	\BlankLine
	\ForEach{request $R(i)$ from user $U(i)$}{
		\ForEach{sensing task $ST_{R(i)}(j)$}{	
				\Recruit{$P_{ij}$ users with highest AVG[] score}\;
				\QoD{Sensing data from $P_{ij}$ users}\;
				\Update{AVG[$P_{ij}$ users]}\;
		}
		$out$ $\leftarrow$ $out$ + \QoS{$R(i)$};
	}
	\textbf{Return} $out$
	\caption{Average-based QoD Recruitment Algorithm}\label{avgbased}
\end{algorithm}

\begin{itemize}
	\item \textbf{Polynomial Regression-based QoD User Recruitment scheme}
\end{itemize}

This scheme maintains a history of QoD scores that $N$ users have contributed to the MCS system and recruits participants based on a prediction on QoD scores for next sensing tasks using a polynomial regression model. The 3-degree polynomial model by means of the least-square fit method is used as the predictive model in the algorithm.

As can be seen in Algorithm \ref{polybased}, it initiates the $QoDScore$ matrix for storing the history of QoD scores in previous sensing tasks for $N$ users (line \#1). The output at the beginning state is set to 0 (line \#2). For each request $R(i)$ from a user $U(i)$ and for each sensing task $ST_{R(i)}(j)$, the algorithm uses the $polyfit$ and $polyval$ functions for finding the coefficients and predicting the next QoD scores for each user (line \#5, line \#6); then, it recruits users with highest predicted QoD scores (line \#7). When the sensing task has been accomplished, the algorithm calculates the QoD score for the sensing data collected from the recruited users (line \#8) and updates the $QoDScore$ matrix accordingly (line \#9). Finally, the output is updated by adding the QoS score of the requested service $R(i)$ (line \#11).

\begin{algorithm}
	\SetKwFunction{Polyfit}{polyfit}
	\SetKwFunction{Polyval}{polyval}
	\SetKwFunction{Update}{Update}
	\SetKwFunction{Recruit}{Recruit}
	\SetKwFunction{QoD}{QoD}
	\SetKwFunction{QoS}{QoS}

	\textbf{Initialization}
		QoDScore[][]; \;
		$out$ = 0; \;
	\BlankLine
	\ForEach{request $R(i)$ from user $U(i)$}{
		\ForEach{sensing task $ST_{R(i)}(j)$}{
				f = \Polyfit{(t, QoDScore[][],3)};\;
				\Polyval{(f, t+1)}; \;
			
				\Recruit{$P_{ij}$ users with highest predicted QoD score}\;
				\QoD{Collected Data from $P_{ij}$ users}\;
				\Update{QoDScore[$P_{ij}$ users]}\;
		}
		$out$ $\leftarrow$ $out$ + \QoS{$R(i)$};
	}
	\textbf{Return} $out$
	\caption{Polynomial Regression-based QoD Recruitment Algorithm}\label{polybased}
\end{algorithm}

\section{SIMULATION RESULTS AND DISCUSSIONS} \label{RESULT}
The testbed is implemented in Matlab containing a set of users consisting of low-quality, high-quality and malicious users, a number of service requests, and the three user recruitment schemes. For comparison purposes, all three schemes take the same inputs (i.e., the set of users and the service requests) and produce outputs as QoS scores for the requested services. The source code for the implementation can be found here\footnote{\url{https://github.com/nguyentb/MCS_project}}.

\subsection{Parameter Settings for Experience Model}
As discussed in Section IV.B, $max_{Exp}$ and $min_{Exp}$ are set to $1$ and $0$, respectively. $Exp_0$ is set to $0.3$ at the bootstrap state. According to the statistics of the QoD scores discussed in Section V.A, if a user provides a dataset with a QoD score $\geq 0.6$ then it is a cooperative interaction; otherwise if the QoD score is $\leq 0.3$ meaning that the user provides a very low-quality dataset, then it is an uncooperative interaction. $\alpha$ is the maximum Increase value and the smaller the $\alpha$ is, the more interactions are required to get a strong relationship. As can be seen in Fig. \ref{fig_6x}, we set $\alpha = 0.1$, as a result, it takes more than $15$ interactions in order to attain a strong relationship (i.e., the Experience value $\geq 0.9$). Similar experiments were conducted to come up with the other controlling parameters for the Decrease model and Decay model (i.e., $\delta$ and $\gamma$) for forming reasonable curves as shown in Fig. \ref{fig_6x}.

\begin{figure}[ht]
\centering
	\includegraphics[width=0.5\textwidth]{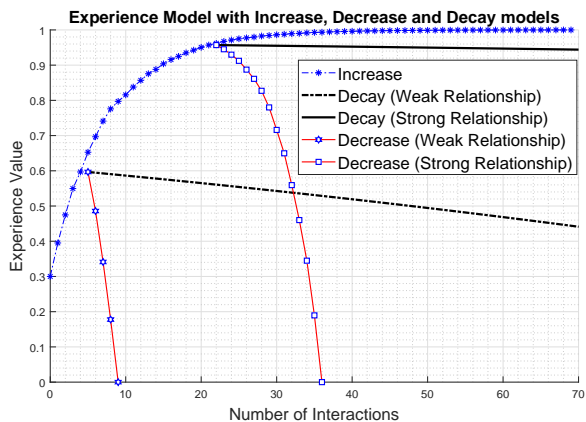}
		\caption{Experience Model with Increase, Decrease and Decay models}
\label{fig_6x}
\end{figure}

Note that different use-cases might result in different parameter settings, depending on how difficult it is to build up a strong relationship as well as to lose and decay the relationship. Details for the parameters used in this article are shown in Table \ref{tb1}.

\begin{table}[h]
	\centering
	\captionsetup{justification=centering}
	\caption{PARAMETERS SETTING FOR THE EXPERIENCE MODEL}
	\label{tb1}
	\begin{tabular}{c|c||c|c}		
			\textbf{Parameters} & \textbf{Values} & \textbf{Parameters} & \textbf{Values} \\ [0.5ex]
		\hline\hline
			$max_{Exp}$ & 1 & $\gamma$ & 0.005 \\
		\hline
			$min_{Exp}$ & 0 & $\delta$ & 0.005 \\
		\hline
			$Exp_0$ & 0.3 & $\theta_{unco}$ & 0.3 \\
		\hline
			$\alpha$ & 0.1 & $\theta_{co}$ & 0.6 \\
		\hline
		$\beta$ & 2 &  & \\
	\end{tabular} \\[0.1ex]
\end{table}

\subsection{Calculation Mechanism for the Reputation Model}

The Reputation mechanism in a MCS system can be calculated either algebraically or iteratively. The traditional algebra method to solve the matrix equations in (\ref{eq_11}) and (\ref{eq_12}) takes roughly $N^3$ operations that is not suitable for a large number of users ($N$ is the network size, i.e., the number of users). On the other hand, the iterative method is much faster because the $Rep_{Pos}$  and $Rep_{Neg}$ vectors converge after conducting a number of iterations \cite{r55}. We therefore use the second method in this simulation and, with the damping factor set to $0.85$, the $error\_tolerance = 10^3$, and the number of users ranging from $200$ to $1000$, it takes from $25$ to $32$ iterations to converge. This reputation calculation is suitable for huge networks like the IoT as the scaling factor is roughly linear in logarithm of $N$ \cite{r07}.

\begin{itemize}
	\item \textbf{Testbed simulation scenarios}
\end{itemize}
The number of service requests is varied from $1$ to $160$, and without the loss of generality, we assume that each service request is fulfilled by a random number of sensing tasks from $5$ to $15$. Each sensing task requires a number of users from $5$ to $200$ (50\% of the total users). The total number of users $N$ is set at $400$; and the number of malicious users is varied from 0\% to 25\% of $N$. We also assume that a user can participate in several tasks simultaneously.

\subsection{Results and Discussion}
We have implemented the three algorithms outlined above in the simulation and, for better observation, we have also implemented a random selection method as the simplest recruitment scheme. As can be seen in Fig. \ref{fig7}, the Trust-based scheme outperforms all other schemes in most of the cases, meaning that the quality of the requested services using the proposed trust-based user recruitment is better than the other schemes. 	All the schemes, except the Random Selection, produce better QoS scores as more requested services are served. However, after a period of about $15$ requests (i.e., the learning phase), the Trust-based scheme achieves consistent QoS scores for following services whereas the Average-based and the Polynomial Regression take about $35$ and $70$ requests, respectively. After the learning phase, the Trust-based scheme persistently achieves the highest QoS scores compared to the other schemes at about $3.35$ to $3.55$, whereas the Average-based scheme fluctuated between $3.10$ and $3.35$ while the Regression outcomes steadily increased and ultimately reaches about $3.25$ to $3.40$.

\begin{figure}[ht]
	\centering
	\captionsetup{justification=centering}
	\includegraphics[width=0.48\textwidth]{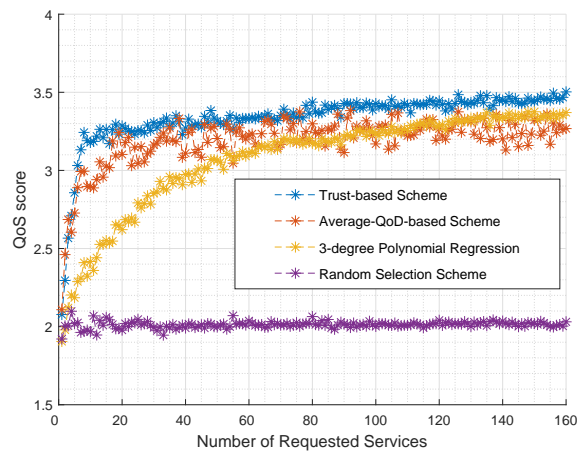}
	\caption{QoS scores after numbers of services using different User Recruitment schemes}
	\label{fig7}
\end{figure}

The three schemes all learn from previous data contributors for maximizing the outcomes. However, with the exception of the Trust-based scheme, they fail to detect malicious users. That is why some malicious users are still recruited in these schemes resulting in lowering the QoS scores for requested services. This is understandable because the Average-based scheme will consider malicious users to be high-quality users due to their average QoD scores being similar. Compared to the Average-based scheme, the Regression method produces just slightly better QoS scores and is more consistent after a long learning phase. This is because malicious users contribute high-quality data most of the time so that low-quality data, which rarely occurs, could be considered as outliers in the regression model. As such, some malicious users are quantified as high-quality users. The regression model also requires more data points for a more accurate prediction, resulting in the longer learning phase.

\begin{figure*}[!ht]
	\centering
	\captionsetup{justification=centering}
	\includegraphics[width=\textwidth]{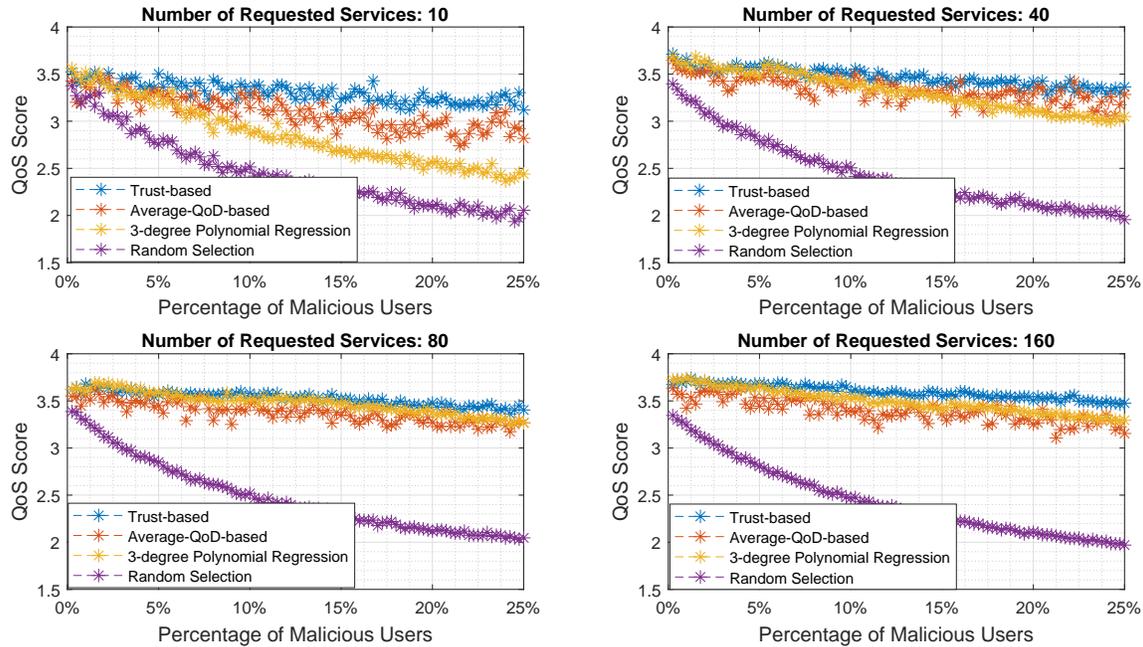}
	\caption{QoS scores in different Percentages of Malicious Users using different User Recruitment Schemes}
	\label{fig8}
\end{figure*}

Unlike these two schemes, the E-R model heavily penalizes a user who sometimes produces very low QoD scores, resulting in rapid drops in the trust relationship and the reputation value of that user. By looking at the reputation vector for all users after the learning phase, we notice that reputation values of malicious users are normally lower than low-quality users and far lower than high-quality users. Considering a scenario in which the number of malicious users is 10\% ($40$ malicious users out of $400$ total users), we examined the users with the lowest reputation values after the learning phase (i.e., after $20$ service requests). As can be seen in Table \ref{tb2}, 80\% of the malicious users are detected just by looking at 10\% of the users having the lowest reputation values. Moreover, in the $80$ users (20\% of the total users) with the lowest reputation values there are $35$ out of $40$ malicious users.
That is why after the learning phase, the trust-based scheme tends to avoid recruiting these malicious users; because there is a very high possibility that a low reputation value results in a low overall trust value.

\begin{table}[h]
	\centering
	\captionsetup{justification=centering}
	\caption{LOWEST REPUTATION VALUES IN ACCORDANCE WITH PERCENTAGE OF USERS TYPES}
	\label{tb2}
	\begin{tabular}{c|c|c|c}
		\textbf{Lowest Reputation} & \textbf{Malicious} & \textbf{Low-Quality} & \textbf{High-Quality} \\ [0.5ex]
		\hline\hline
		10 users (2.5\%) & 10 (100\%) & 0 (0\%) & 0 (0\%) \\
		\hline
		20 users (5\%) & 19 (95\%) & 1 (5\%) & 0 (0\%) \\
		\hline
		30 users (7.5\%) & 26 (87\%) & 4 (13\%) & 0 (0\%) \\
		\hline
		40 users (10\%) & 32 (80\%) & 8 (20\%) & 0 (0\%) \\
		80 users (20\%) & 35 & 43 & 2 \\
	\end{tabular} \\[0.1ex]
\end{table}

We also examined scenarios in which the number of malicious users are varied. As shown in Fig. \ref{fig8}, the percentage of malicious users over total users is increased (i.e, from 0\% to 25\% of the total number of users), the QoS is also decreased (in all scenarios with different numbers of requested services (i.e., 10, 40, 80, 160 requested services)). This is inevitable because the possibility of recruiting malicious users is higher. However, as the number of requested services increase, the QoS scores from all recruitment schemes, except the Random Selection, get higher. For instance, at 15\% of malicious users, the QoS scores from the Trust-based scheme increased from about $3.2$, $3.35$, $3.5$ and $3.6$ after serving $10$, $40$, $80$ and $160$ services, respectively.

As can also be seen in Fig. \ref{fig8}, as the number of malicious users increase, the gap in QoS scores between the Trust-based scheme and the others schemes expands, especially when more requested services are served, showing the advantages of the Trust-based scheme in untrustworthy environments. For example, when the number of requested services is $160$ (as shown in the below-right subplot of Fig. \ref{fig8}), with 10\% of malicious users, the QoS scores obtained from the Trust-based scheme and the Regression scheme are $3.65$ and $3.58$, respectively; with 25\% of malicious users, the QoS scores are $3.49$ and $3.31$. Therefore, the QoS score gap between the proposed Trust-based and the Regression schemes increases from $3.65-3.58=0.07$ to $3.49-3.31=0.18$.

If the percentage of malicious users is less than 10\%, the Average-based scheme seems the best option which offers similar QoS scores but requires less computing resources. Unlike the Experience model, the Reputation model requires significant computational resources. Thus, it is not necessarily desirable to execute the Reputation mechanism in every evaluation of trust. In reality, the reputation mechanism should be periodically performed, which could drastically save time and computational resources.
\section{Conclusion and Future Work} \label{CFW}
In this article we propose a trust evaluation mechanism to create and maintain trust relationships between mobile device users in a MCS platform called E-R. To establish and manage the trust relationships, we introduce the concept of virtual interactions in a centralized MCS platform, forming when a user contributes data for a sensing task from a service requester. Such interactions are quantified by the assessment of the quality of the contributed data; and used as the inputs for the calculation of the two indicators of trust: Experience and Reputation; and the trust relationships between the MCS users are attained by incorporating these two TIs. Based on the trust relationships, a trust-based user recruitment scheme in a MCS platform is proposed for selecting the most trustworthy users with the purpose of contributing high-quality data.

In order to show the effectiveness of the proposed E-R trust mechanism and the trust-based user recruitment, we simulate a MCS testbed consisting of both normal and malicious users with the deployment of the trust-based recruitment scheme along with three other recruitment mechanisms for comparison. The results reveal that the trust-based scheme outperforms the other schemes as it provides better QoS for MCS services in most cases. The trust-based scheme is also able to envisage different types of users including intelligent malicious users, preventing the from being recruited for sensing tasks. Moreover, the proposed recruitment mechanisms is practically implemented in real-world IoT services as we have done in the Wise-IoT project\footnote{\url{http://wise-iot.eu/2018/03/29/march-2018-8}}, which is also an achievement over other recruitment mechanisms which rely on unrealistic assumptions.

This article opens some future research directions. The first direction is an automatic adaptation of parameter settings for the Experience and Reputation models in a context-aware manner. Different MCS systems have different characteristics and types of users which need to be examined, meaning that the QoD assessment, the user models and the QoS evaluation model could also be different. This opens a second research direction for customizing the proposed mechanism for specific MCS use-cases. For example, the trustworthiness of data contributors can also be used in a crowd-sensing data model for better handling of noisy and unreliable data from mobile users, which could effectively improve the data quality in MCS systems \cite{r1_3}. A fourth direction is the integration of the Knowledge TI that contains various useful information of MCS systems. This could result in even more precise indications of trustworthy mobile users; or the integration of other mechanisms like Incentive for a better recruitment scheme.

\section*{Acknowledgment}
This research was supported by the Basic Science Research Program through the National Research Foundation of Korea (NRF) funded by the Ministry of Science and ICT (2018R1A2B2003774).

\bibliographystyle{IEEEtran}
\bibliography{biblio}

\begin{thebibliography}{10}
\providecommand{\url}[1]{#1}
\csname url@samestyle\endcsname
\providecommand{\newblock}{\relax}
\providecommand{\bibinfo}[2]{#2}
\providecommand{\BIBentrySTDinterwordspacing}{\spaceskip=0pt\relax}
\providecommand{\BIBentryALTinterwordstretchfactor}{4}
\providecommand{\BIBentryALTinterwordspacing}{\spaceskip=\fontdimen2\font plus
\BIBentryALTinterwordstretchfactor\fontdimen3\font minus
  \fontdimen4\font\relax}
\providecommand{\BIBforeignlanguage}[2]{{%
\expandafter\ifx\csname l@#1\endcsname\relax
\typeout{** WARNING: IEEEtran.bst: No hyphenation pattern has been}%
\typeout{** loaded for the language `#1'. Using the pattern for}%
\typeout{** the default language instead.}%
\else
\language=\csname l@#1\endcsname
\fi
#2}}
\providecommand{\BIBdecl}{\relax}
\BIBdecl

\bibitem{r01}
W.~Z. Khan, Y.~Xiang, M.~Y. Aalsalem, and Q.~Arshad, ``Mobile phone sensing
  systems: A survey,'' \emph{IEEE Communications Surveys \& Tutorials},
  vol.~15, no.~1, pp. 402--427, 2013.

\bibitem{r02}
R.~K. Ganti, F.~Ye, and H.~Lei, ``Mobile crowdsensing: current state and future
  challenges,'' \emph{IEEE Communications Magazine}, vol.~49, no.~11, 2011.

\bibitem{r03}
B.~Guo, C.~Chen, D.~Zhang, Z.~Yu, and A.~Chin, ``Mobile crowd sensing and
  computing: when participatory sensing meets participatory social media,''
  \emph{IEEE Communications Magazine}, vol.~54, no.~2, pp. 131--137, 2016.

\bibitem{r04}
B.~Guo, Z.~Wang, Z.~Yu, Y.~Wang, N.~Y. Yen, R.~Huang, and X.~Zhou, ``Mobile
  crowd sensing and computing: The review of an emerging human-powered sensing
  paradigm,'' \emph{ACM Computing Surveys (CSUR)}, vol.~48, no.~1, pp. 7--17,
  2015.

\bibitem{r05}
G.~Ding, J.~Wang, Q.~Wu, L.~Zhang, Y.~Zou, Y.-D. Yao, and Y.~Chen, ``Robust
  spectrum sensing with crowd sensors,'' \emph{IEEE Transactions on
  Communications}, vol.~62, no.~9, pp. 3129--3143, 2014.

\bibitem{r06}
N.~B. Truong, H.~Lee, B.~Askwith, and G.~M. Lee, ``Toward a trust evaluation
  mechanism in the social internet of things,'' \emph{Sensors}, vol.~17, no.~6,
  pp. 1346--1360, 2017.

\bibitem{r07}
N.~B. Truong, T.-W. Um, B.~Zhou, and G.~M. Lee, ``From personal experience to
  global reputation for trust evaluation in the social internet of things,'' in
  \emph{IEEE Global Communications Conference (GLOBECOM)}.\hskip 1em plus 0.5em
  minus 0.4em\relax IEEE, 2017, pp. 1--7.

\bibitem{r08}
J.~Liu and W.~Sun, ``Smart attacks against intelligent wearables in
  people-centric internet of things,'' \emph{IEEE Communications Magazine},
  vol.~54, no.~12, pp. 44--49, 2016.

\bibitem{r09}
J.~Gubbi, R.~Buyya, S.~Marusic, and M.~Palaniswami, ``Internet of things
  {(IoT)}: A vision, architectural elements, and future directions,''
  \emph{Future generation computer systems}, vol.~29, no.~7, pp. 1645--1660,
  2013.

\bibitem{r10}
O.~Vermesan, P.~Friess, P.~Guillemin, S.~Gusmeroli, H.~Sundmaeker, A.~Bassi,
  I.~S. Jubert, M.~Mazura, M.~Harrison, M.~Eisenhauer \emph{et~al.}, ``Internet
  of things strategic research roadmap,'' \emph{Internet of Things-Global
  Technological and Societal Trends}, vol.~1, no. 2011, pp. 9--52, 2011.

\bibitem{r11}
J.~S. Silva, P.~Zhang, T.~Pering, F.~Boavida, T.~Hara, and N.~C. Liebau,
  ``People-centric internet of things,'' \emph{IEEE Communications Magazine},
  vol.~55, no.~2, pp. 18--19, 2017.

\bibitem{r12}
B.~Guo, Z.~Yu, X.~Zhou, and D.~Zhang, ``From participatory sensing to mobile
  crowd sensing,'' in \emph{Pervasive Computing and Communications Workshops
  (PERCOM Workshops), IEEE International Conference on}.\hskip 1em plus 0.5em
  minus 0.4em\relax IEEE, 2014, pp. 593--598.

\bibitem{r13}
F.~Delmastro, V.~Arnaboldi, and M.~Conti, ``People-centric computing and
  communications in smart cities,'' \emph{IEEE Communications Magazine},
  vol.~54, no.~7, pp. 122--128, 2016.

\bibitem{r14}
A.~Capponi, C.~Fiandrino, D.~Kliazovich, P.~Bouvry, and S.~Giordano, ``A
  cost-effective distributed framework for data collection in cloud-based
  mobile crowd sensing architectures,'' \emph{IEEE Transactions on Sustainable
  Computing}, vol.~2, no.~1, pp. 3--16, 2017.

\bibitem{r1_4}
B.~Zhang, C.~H. Liu, J.~Lu, Z.~Song, Z.~Ren, J.~Ma, and W.~Wang,
  ``Privacy-preserving qoi-aware participant coordination for mobile
  crowdsourcing,'' \emph{Computer Networks}, vol. 101, pp. 29--41, 2016.

\bibitem{r15}
M.-R. Ra, B.~Liu, T.~F. La~Porta, and R.~Govindan, ``Medusa: A programming
  framework for crowd-sensing applications,'' in \emph{Proceedings of the 10th
  international conference on Mobile systems, applications, and services
  (MobiSys)}.\hskip 1em plus 0.5em minus 0.4em\relax ACM, 2012, pp. 337--350.

\bibitem{r16}
D.~Zhang, L.~Wang, H.~Xiong, and B.~Guo, ``4w1h in mobile crowd sensing,''
  \emph{IEEE Communications Magazine}, vol.~52, no.~8, pp. 42--48, 2014.

\bibitem{r5_3}
D.~Mendez, M.~Labrador, and K.~Ramachandran, ``Data interpolation for
  participatory sensing systems,'' \emph{Pervasive and Mobile Computing},
  vol.~9, no.~1, pp. 132--148, 2013.

\bibitem{r17}
S.~Reddy, D.~Estrin, and M.~Srivastava, ``Recruitment framework for
  participatory sensing data collections,'' in \emph{International Conference
  on Pervasive Computing}.\hskip 1em plus 0.5em minus 0.4em\relax Springer,
  2010, pp. 138--155.

\bibitem{r18}
M.~Karaliopoulos, O.~Telelis, and I.~Koutsopoulos, ``User recruitment for
  mobile crowdsensing over opportunistic networks,'' in \emph{Computer
  Communications (INFOCOM), 2015 IEEE Conference on}.\hskip 1em plus 0.5em
  minus 0.4em\relax IEEE, 2015, pp. 2254--2262.

\bibitem{r19}
D.~Zhang, H.~Xiong, L.~Wang, and G.~Chen, ``Crowdrecruiter: selecting
  participants for piggyback crowdsensing under probabilistic coverage
  constraint,'' in \emph{Proceedings of the 2014 ACM International Joint
  Conference on Pervasive and Ubiquitous Computing}.\hskip 1em plus 0.5em minus
  0.4em\relax ACM, 2014, pp. 703--714.

\bibitem{r20}
N.~D. Lane, Y.~Chon, L.~Zhou, Y.~Zhang, F.~Li, D.~Kim, G.~Ding, F.~Zhao, and
  H.~Cha, ``Piggyback crowdsensing (pcs): energy efficient crowdsourcing of
  mobile sensor data by exploiting smartphone app opportunities,'' in
  \emph{Proceedings of the 11th ACM Conference on Embedded Networked Sensor
  Systems (SenSys)}.\hskip 1em plus 0.5em minus 0.4em\relax ACM, 2013, pp.
  1--7.

\bibitem{r21}
C.~H. Liu, B.~Zhang, X.~Su, J.~Ma, W.~Wang, and K.~K. Leung, ``Energy-aware
  participant selection for smartphone-enabled mobile crowd sensing,''
  \emph{IEEE Systems Journal}, vol.~11, no.~3, pp. 1435--1446, 2017.

\bibitem{r23}
C.~Fiandrino, B.~Kantarci, F.~Anjomshoa, D.~Kliazovich, P.~Bouvry, and
  J.~Matthews, ``Sociability-driven user recruitment in mobile crowdsensing
  internet of things platforms,'' in \emph{Global Communications Conference
  (GLOBECOM), IEEE}.\hskip 1em plus 0.5em minus 0.4em\relax IEEE, 2016, pp.
  1--6.

\bibitem{minor_1}
J.~Wang, F.~Wang, Y.~Wang, D.~Zhang, L.~Wang, and Z.~Qiu,
  ``Social-network-assisted worker recruitment in mobile crowd sensing,''
  \emph{arXiv preprint arXiv:1805.08525}, 2018.

\bibitem{r5_6}
E.~Wang, Y.~Yang, J.~Wu, W.~Liu, and X.~Wang, ``An efficient prediction-based
  user recruitment for mobile crowdsensing,'' \emph{IEEE Transactions on Mobile
  Computing}, vol.~17, no.~1, pp. 16--28, 2018.

\bibitem{r5_1}
H.~Li, T.~Li, and Y.~Wang, ``Dynamic participant recruitment of mobile crowd
  sensing for heterogeneous sensing tasks,'' in \emph{Mobile Ad Hoc and Sensor
  Systems (MASS), 2015 IEEE 12th International Conference on}.\hskip 1em plus
  0.5em minus 0.4em\relax IEEE, 2015, pp. 136--144.

\bibitem{r22}
S.~He, D.-H. Shin, J.~Zhang, and J.~Chen, ``Toward optimal allocation of
  location dependent tasks in crowdsensing,'' in \emph{INFOCOM, Proceedings
  IEEE}.\hskip 1em plus 0.5em minus 0.4em\relax IEEE, 2014, pp. 745--753.

\bibitem{r1_1}
C.~H. Liu, J.~Fan, P.~Hui, J.~Wu, and K.~K. Leung, ``Toward qoi and energy
  efficiency in participatory crowdsourcing,'' \emph{IEEE Transactions on
  Vehicular Technology}, vol.~64, no.~10, pp. 4684--4700, 2015.

\bibitem{minor_3}
W.~Li, F.~Li, K.~Sharif, and Y.~Wang, ``When user interest meets data quality:
  A novel user filter scheme for mobile crowd sensing,'' in \emph{Parallel and
  Distributed Systems (ICPADS), 2017 IEEE 23rd International Conference
  on}.\hskip 1em plus 0.5em minus 0.4em\relax IEEE, 2017, pp. 97--104.

\bibitem{r5_2}
C.~H. Liu, B.~Zhang, X.~Su, J.~Ma, W.~Wang, and K.~K. Leung, ``Energy-aware
  participant selection for smartphone-enabled mobile crowd sensing,''
  \emph{IEEE Systems Journal}, vol.~11, no.~3, pp. 1435--1446, 2017.

\bibitem{r5_4}
H.~Jin, L.~Su, D.~Chen, K.~Nahrstedt, and J.~Xu, ``Quality of information aware
  incentive mechanisms for mobile crowd sensing systems,'' in \emph{Proceedings
  of the 16th ACM International Symposium on Mobile Ad Hoc Networking and
  Computing}.\hskip 1em plus 0.5em minus 0.4em\relax ACM, 2015, pp. 167--176.

\bibitem{r24}
D.~Yang, G.~Xue, X.~Fang, and J.~Tang, ``Crowdsourcing to smartphones:
  Incentive mechanism design for mobile phone sensing,'' in \emph{Proceedings
  of the 18th annual international conference on Mobile computing and
  networking}.\hskip 1em plus 0.5em minus 0.4em\relax ACM, 2012, pp. 173--184.

\bibitem{r25}
B.~Kantarci and H.~T. Mouftah, ``Trustworthy sensing for public safety in
  cloud-centric internet of things,'' \emph{IEEE Internet of Things Journal},
  vol.~1, no.~4, pp. 360--368, 2014.

\bibitem{r26}
B.~Kantarci, K.~G. Carr, and C.~D. Pearsall, ``Sonata: Social network assisted
  trustworthiness assurance in smart city crowdsensing,'' \emph{International
  Journal of Distributed Systems and Technologies (IJDST)}, vol.~7, no.~1, pp.
  59--78, 2016.

\bibitem{r27}
M.~Pouryazdan and B.~Kantarci, ``The smart citizen factor in trustworthy smart
  city crowdsensing,'' \emph{IT Professional}, vol.~18, no.~4, pp. 26--33,
  2016.

\bibitem{minor_5}
M.~Pouryazdan, C.~Fiandrino, B.~Kantarci, T.~Soyata, D.~Kliazovich, and
  P.~Bouvry, ``Intelligent gaming for mobile crowd-sensing participants to
  acquire trustworthy big data in the internet of things,'' \emph{IEEE Access},
  vol.~5, pp. 22\,209--22\,223, 2017.

\bibitem{minor_6}
M.~Pouryazdan, C.~Fiandrino, B.~Kantarci, D.~Kliazovich, T.~Soyata, and
  P.~Bouvry, ``Game-theoretic recruitment of sensing service providers for
  trustworthy cloud-centric internet-of-things (iot) applications,'' in
  \emph{IEEE Global Communications Conference (GLOBECOM) Workshops: Fifth
  International Workshop on Cloud Computing Systems, Networks, and Applications
  (CCSNA)}, 2016.

\bibitem{r29}
B.~Kantarci and H.~T. Mouftah, ``Trustworthy crowdsourcing via mobile social
  networks,'' in \emph{Global Communications Conference (GLOBECOM), 2014
  IEEE}.\hskip 1em plus 0.5em minus 0.4em\relax IEEE, 2014, pp. 2905--2910.

\bibitem{r1_2}
Y.~Gao, X.~Li, J.~Li, and Y.~Gao, ``A dynamic-trust-based recruitment framework
  for mobile crowd sensing,'' in \emph{Communications (ICC), 2017 IEEE
  International Conference on}.\hskip 1em plus 0.5em minus 0.4em\relax IEEE,
  2017, pp. 1--6.

\bibitem{r5_5}
F.~Restuccia, N.~Ghosh, S.~Bhattacharjee, S.~K. Das, and T.~Melodia, ``Quality
  of information in mobile crowdsensing: Survey and research challenges,''
  \emph{ACM Transactions on Sensor Networks (TOSN)}, vol.~13, no.~4, pp.
  34--44, 2017.

\bibitem{r30}
P.-Y. Chen, S.-M. Cheng, P.-S. Ting, C.-W. Lien, and F.-J. Chu, ``When
  crowdsourcing meets mobile sensing: a social network perspective,''
  \emph{IEEE Communications Magazine}, vol.~53, no.~10, pp. 157--163, 2015.

\bibitem{r31}
J.~Weppner, P.~Lukowicz, U.~Blanke, and G.~Tr{\"o}ster, ``Participatory
  bluetooth scans serving as urban crowd probes,'' \emph{IEEE Sensors Journal},
  vol.~14, no.~12, pp. 4196--4206, 2014.

\bibitem{r32}
R.~Zhang, J.~Zhang, Y.~Zhang, J.~Sun, and G.~Yan, ``Privacy-preserving profile
  matching for proximity-based mobile social networking,'' \emph{IEEE Journal
  on Selected Areas in Communications}, vol.~31, no.~9, pp. 656--668, 2013.

\bibitem{r33}
C.~Perera, A.~Zaslavsky, P.~Christen, and D.~Georgakopoulos, ``Sensing as a
  service model for smart cities supported by internet of things,''
  \emph{Transactions on Emerging Telecommunications Technologies}, vol.~25,
  no.~1, pp. 81--93, 2014.

\bibitem{r34}
G.~Merlino, S.~Arkoulis, S.~Distefano, C.~Papagianni, A.~Puliafito, and
  S.~Papavassiliou, ``Mobile crowdsensing as a service: a platform for
  applications on top of sensing clouds,'' \emph{Future Generation Computer
  Systems}, vol.~56, pp. 623--639, 2016.

\bibitem{r35}
N.~B. Truong, Q.~H. Cao, T.-W. Um, and G.~M. Lee, ``Leverage a trust service
  platform for data usage control in smart city,'' in \emph{Global
  Communications Conference (GLOBECOM), 2016 IEEE}.\hskip 1em plus 0.5em minus
  0.4em\relax IEEE, 2016, pp. 1--7.

\bibitem{r36}
D.~Christin, ``Privacy in mobile participatory sensing: Current trends and
  future challenges,'' \emph{Journal of Systems and Software}, vol. 116, pp.
  57--68, 2016.

\bibitem{r37}
N.~Laranjeiro, S.~N. Soydemir, and J.~Bernardino, ``A survey on data quality:
  classifying poor data,'' in \emph{Dependable Computing (PRDC), 2015 IEEE 21st
  Pacific Rim International Symposium on}.\hskip 1em plus 0.5em minus
  0.4em\relax IEEE, 2015, pp. 179--188.

\bibitem{r38}
N.~Askham, D.~Cook, M.~Doyle, H.~Fereday, M.~Gibson, U.~Landbeck, R.~Lee,
  C.~Maynard, G.~Palmer, and J.~Schwarzenbach, ``The six primary dimensions for
  data quality assessment,'' \emph{DAMA UK Working Group}, pp. 432--435, 2013.

\bibitem{r40}
D.~Loshin, ``Data quality assessment,'' in \emph{The practitioner's guide to
  data quality improvement}.\hskip 1em plus 0.5em minus 0.4em\relax Elsevier,
  2010, ch.~11, pp. 191--206.

\bibitem{r41}
L.~L. Pipino, Y.~W. Lee, and R.~Y. Wang, ``Data quality assessment,''
  \emph{Communications of the ACM}, vol.~45, no.~4, pp. 211--218, 2002.

\bibitem{r42}
H.~Kim, A.~Ahmad, J.~Hwang, H.~Baqa, F.~Le~Gall, M.~A.~R. Ortega, and J.~Song,
  ``{IoT-TaaS}: Towards a prospective {IoT} testing framework,'' \emph{IEEE
  Access}, vol.~6, pp. 15\,480--15\,493, 2018.

\bibitem{r45}
S.~Kraounakis, I.~N. Demetropoulos, A.~Michalas, M.~S. Obaidat, P.~G.
  Sarigiannidis, and M.~D. Louta, ``A robust reputation-based computational
  model for trust establishment in pervasive systems,'' \emph{IEEE Systems
  Journal}, vol.~9, no.~3, pp. 878--891, 2015.

\bibitem{r46}
K.~H. Rubin, W.~M. Bukowski, and B.~Laursen, \emph{Handbook of peer
  interactions, relationships, and groups}.\hskip 1em plus 0.5em minus
  0.4em\relax Guilford Press, 2011.

\bibitem{r47}
M.~Demir, M.~{\"O}zdemir, and K.~P. Marum, ``Perceived autonomy support,
  friendship maintenance, and happiness,'' \emph{The Journal of psychology},
  vol. 145, no.~6, pp. 537--571, 2011.

\bibitem{r48}
S.~G. Roberts, R.~I. Dunbar, T.~V. Pollet, and T.~Kuppens, ``Exploring
  variation in active network size: Constraints and ego characteristics,''
  \emph{Social Networks}, vol.~31, no.~2, pp. 138--146, 2009.

\bibitem{r49}
S.~Brin and L.~Page, ``Reprint of: The anatomy of a large-scale hypertextual
  web search engine,'' \emph{Computer networks}, vol.~56, no.~18, pp.
  3825--3833, 2012.

\bibitem{r50}
N.~Tyagi and S.~Sharma, ``Weighted page rank algorithm based on number of
  visits of links of web page,'' \emph{International Journal of Soft Computing
  and Engineering (IJSCE) ISSN}, pp. 2231--2307, 2012.

\bibitem{r6_18}
S.~Brin and L.~Page, ``Reprint of: The anatomy of a large-scale hypertextual
  web search engine,'' \emph{Computer networks}, vol.~56, no.~18, pp.
  3825--3833, 2012.

\bibitem{r6_29}
W.~Xie, D.~Bindel, A.~Demers, and J.~Gehrke, ``Edge-weighted personalized
  pagerank: Breaking a decade-old performance barrier,'' in \emph{Proceedings
  of the 21th ACM SIGKDD International Conference on Knowledge Discovery and
  Data Mining}.\hskip 1em plus 0.5em minus 0.4em\relax ACM, 2015, pp.
  1325--1334.

\bibitem{r6_26}
N.~Tyagi and S.~Sharma, ``Weighted page rank algorithm based on number of
  visits of links of web page,'' \emph{International Journal of Soft Computing
  and Engineering (IJSCE) ISSN}, pp. 2231--2307, 2012.

\bibitem{r6_28}
L.~Backstrom and J.~Leskovec, ``Supervised random walks: predicting and
  recommending links in social networks,'' in \emph{Proceedings of the fourth
  ACM international conference on Web search and data mining}.\hskip 1em plus
  0.5em minus 0.4em\relax ACM, 2011, pp. 635--644.

\bibitem{r51}
S.~Geisser, \emph{Predictive inference}.\hskip 1em plus 0.5em minus 0.4em\relax
  Routledge, 2017.

\bibitem{r52}
C.~Batini, C.~Cappiello, C.~Francalanci, and A.~Maurino, ``Methodologies for
  data quality assessment and improvement,'' \emph{ACM computing surveys
  (CSUR)}, vol.~41, no.~3, p.~16, 2009.

\bibitem{r54}
G.~White, V.~Nallur, and S.~Clarke, ``Quality of service approaches in {IoT}: A
  systematic mapping,'' \emph{Journal of Systems and Software}, vol. 132, pp.
  186--203, 2017.

\bibitem{r55}
M.~Franceschet, ``Pagerank: Standing on the shoulders of giants,''
  \emph{Communications of the ACM}, vol.~54, no.~6, pp. 92--101, 2011.

\bibitem{r1_3}
X.~Ma, Z.~Zheng, F.~Wu, and G.~Chen, ``Trust-based time series data model for
  mobile crowdsensing,'' in \emph{Communications (ICC), 2017 IEEE International
  Conference on}.\hskip 1em plus 0.5em minus 0.4em\relax IEEE, 2017, pp. 1--6.

\end{thebibliography}

\begin{IEEEbiography}[{\includegraphics[width=1in,height=1.25in,clip,keepaspectratio]{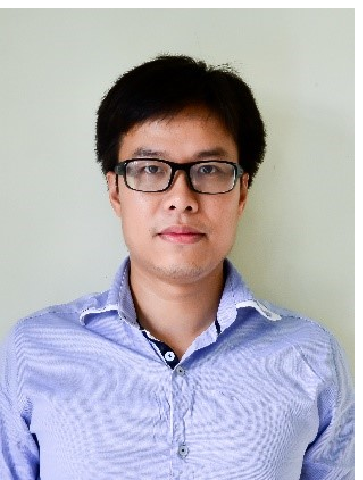}}]{Nguyen Binh Truong}
Dr. Nguyen B.Truong is currently a Research Associate at Data Science Institute, Department of Computing, Imperial College London, United Kingdom. He received his PhD. and Master degrees from Liverpool John Moores University, United Kingdom and Pohang University of Science and Technology, Korea in 2018 and 2013, respectively. He was a Software Engineer at DASAN Networks, a leading company on Networking Products and Services in South Korea from 2012 to 2015. His research interest is including, but not limited to, Security, Privacy and Trust for IoT, Blockchain, Personal Data Management, Fog, Edge and Cloud Computing.
\end{IEEEbiography}

\begin{IEEEbiography}[{\includegraphics[width=1in,height=1.25in,clip,keepaspectratio]{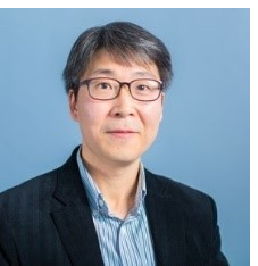}}]{Gyu Myoung Lee}
Dr. Gyu Myoung Lee received his BS degree in electronic and electrical engineering from Hong Ik University, Seoul, Rep. of Korea, in 1999 and MS, and PhD. degree from the Korea Advanced Institute of Science and Technology (KAIST), Daejeon, Rep. of Korea, in 2000 and 2007, respectively. He is currently a Reader in the Department of Computer Science at Liverpool John Moores University, Liverpool (LJMU), UK. He is also with KAIST Institute for IT convergence, Daejeon, Rep. of Korea, as an adjunct professor. His research interests include future networks, Internet of Things, multimedia services, and energy saving networks including Smart Grid. He has actively contributed to standardization in ITU-T as a Rapporteur (currently Q16/13 and Q4/20), oneM2M and IETF. He is also the chair of the ITU-T Focus Group on data processing and management to support IoT and Smart Cities \& Communities.
\end{IEEEbiography}

\begin{IEEEbiography}[{\includegraphics[width=1in,height=1.25in,clip,keepaspectratio]{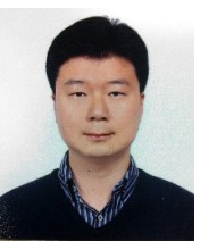}}]{Tai-Won Um}
Dr. Tai-Won Um received the BS degree in electronic and electrical engineering from Hong Ik University, Seoul, Korea, in 1999, and the MS and PhD degrees from the Korea Advanced Institute of Science and Technology (KAIST), Daejeon, Rep. of Korea, in 2000 and 2006, respectively. He is currently an associate professor with Chosun University, Gwangju, Korea. He has been actively participating in standardization meetings including ITU-T SG 13 (Future Networks including mobile, cloud computing and NGN).
\end{IEEEbiography}

\begin{IEEEbiography}[{\includegraphics[width=1in,height=1.25in,clip,keepaspectratio]{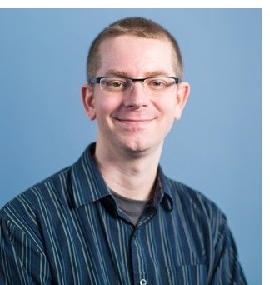}}]{Michael Mackay}
Dr. Michael Mackay is a Senior Lecturer in Computing at LJMU and a member of the PROTECT Research Centre. Michael received a BSc (hons) in Computer Science from Lancaster University in 2000 and a PhD in IPv6 transitioning from the same institution in 2005. Upon completion of his doctoral degree, he was employed as a Senior Research Associate at Lancaster. Dr Mackay has been involved in a number of EU research projects including 6NET (2004), ENTHRONE 2 (2006), EC-GIN (2008), and Wi-5 (2018) in addition to other work funded by the EPSRC, JISC (UK), and Innovate UK. He is also widely published in a wide range of research areas focused around networking technologies including Wireless Communications and Future Internet, QoS and content delivery, and more recently Cloud Computing and the Internet of Things.
\end{IEEEbiography}

\end{document}